\documentclass[12pt,amsfonts]{article}
\usepackage{lineno,hyperref}
\modulolinenumbers[5]
\usepackage{tgtermes}
\usepackage{amsmath, amsfonts, amssymb, amsthm, graphicx}
\usepackage{amsmath,amssymb}
\usepackage{amssymb}
\usepackage{amsmath}
\usepackage{multirow}
\usepackage{graphics}
\usepackage{latexsym}
\usepackage{amsfonts}
\usepackage{xcolor}
\usepackage{epstopdf}
\usepackage{floatrow}
\usepackage{graphicx}
\usepackage{verbatim}
\usepackage[top=3cm , bottom=3cm , left=3cm ,right=3cm ]{geometry}
\newtheorem{theo}{Theorem}[section]

\newtheorem{cor}[theo]{Corollary}
\newtheorem{prop}[theo]{Proposition}
\newtheorem{uppg}[theo]{Exercise}
\newtheorem{remark}[theo]{Remark}
\newtheorem{remarks}[theo]{Remarks}

\newtheorem{definition}[theo]{Definition}
\newcommand{\be}{\begin{eqnarray*}}
	\newcommand{\ee}{\end{eqnarray*}}
\newcommand{\ben}{\begin{eqnarray}}
\newcommand{\een}{\end{eqnarray}}

\def\subsecn (#1) {\medskip\ \ \ {\it #1}\medskip}

\newcommand{\lp}[1]{\left(\begin{array}{#1}}
	\newcommand{\rp}{\end{array}\right)}

\newcommand{\leftd}[1]{\left\{\begin{array}{#1}}
	\newcommand{\rightd}{\end{array}\right.}

\def\B {\mathbf{B}}

\def\K {\mathbf{K}}

\def\P {\mathbf{P}}
\def\Q {\mathbf{Q}}
\def\R {\mathbf{R}}

\def\U {\mathbf{U}}

\def\b {\boldsymbol{b}}

\def\f {\boldsymbol{f}}

\def\k {\boldsymbol{k}}

\def\p {\boldsymbol{p}}

\def\s {\boldsymbol{s}}
\def\t {\boldsymbol{t}}
\def\u {\boldsymbol{u}}

\def\w {\boldsymbol{w}}

\def\z {\boldsymbol{z}}

\def\Ec {\mathcal{E}}

\def\Ic {\mathcal{I}}

\def\Wc {\mathcal{W}}

\def\Rb {\mathbb{R}}

\begin{document}

    \renewcommand{\thefootnote}{\arabic{footnote}}

    \begin{center}
        {\Large \textbf{Interpolation and linear prediction of data - three kernel selection criteria}} \\[0pt]
        ~\\[0pt]
        Azzouz Dermoune\footnote{Laboratoire Paul Painlev\'e,
            USTL-UMR-CNRS 8524. UFR de Math\'ematiques, B\^at. M2. 59655
            Villeneuve d'Ascq C\'edex, France. Email:
            \texttt{azzouz.dermoune@univ-lille1.fr}}, 
             Mohammed Es.Sebaiy\footnote{National School of Applied Sciences-Marrakech, Cadi Ayyad University, Marrakesh, Morocco. E-mail:
            \texttt{mohammedsebaiy@gmail.com}}
            and Jabrane Moustaaid\footnote{National School of Applied Sciences-Marrakech, Cadi Ayyad University, Marrakesh, Morocco.
            E-mail: \texttt{jabrane.mst@gmail.com}}%
        \\[0pt]
        \textit{  Lille University and Cadi Ayyad University }\\[0pt]
        ~\\[0pt]
    \end{center}

    \begin{abstract}
    	Interpolation and prediction have been useful approaches in modeling data in many areas of applications. The aim of this paper is the prediction of the next value of a time series (time series forecasting) using the techniques in interpolation of the spatial data, for the two approaches kernel interpolation and kriging. We are interested in finding some sufficient conditions for the kernels and provide a detailed analyse of the prediction using kernel interpolation. Finally, we provide a natural idea to	select a good kernel among a given family of kernels using only the data. We illustrate our results by application to the data set on the mean annual temperature of France and Morocco recorded for a period of 115 years (1901 to 2015).
    \end{abstract}


\section{Introduction}
Interpolation and prediction have been useful approaches in modelling
data in many areas of applications such as the prediction of the
meteorological variables,  surface reconstruction and Interpolation of spatial data \cite{Scheuerer} among many more. For more details see \cite{Chiles}, \cite{Chiles2},	 \cite{Wandland}  and \cite{MKK}.\\ 
In this work we extend the results of \textit{Scheuerer} \cite{Scheuerer} to the linear prediction approach  of time series. We also cite the work of \textit{Dermoune et all} \cite{DP1} where the parametrizations and the cubic spline were used as a model of prediction and we extend this results to the kernel interpolation framework.

Interpolation of spatial data is a very general mathematical problem and it's precise mathematical formulation as defined in \cite{Scheuerer} is to reconstruct a function $f:T\to \Rb$ with $T$ is	is a domain in $\Rb^d$, based on its values at a finite set of data points $X=\{x_1,\hdots, x_n\}\subset T$, the values $f(x_1),\ldots,f(x_n)$ assumed to be known. But, in our case we are interested in the time series forecasting problem we have $T=\{x_1,\hdots, x_n, x_{n+1}\}$ represent the time and the time series is $f(x_1),\ldots,f(x_n)$ with the unknown value is  $f(x_{n+1})$. In other words,	we want to predict effectively the value	$f(x_{n+1})$ using the known values $f(x_1)$, $\hdots$, $f(x_n)$.
From \cite{Scheuerer} we have that both approaches kernel interpolation and kriging  have the same approximant for the interpolation of spatial data problem, even with the different model assumption, a general overview in both approaches can be fond in \cite{Berlinet}.

\section{Linear prediction and kernel interpolation}
Let $\Rb^{\{x_1,\hdots, x_{n+1}\}}$ be the Hilbert space of real functions on $\{x_1,\hdots, x_{n+1}\}$ with inner product $(.,.)$ and  norm $N(.)$. The dual of $\Rb^{\{x_1,\hdots, x_{n+1}\}}$ is spanned by the point evaluation linear
forms $\delta_x: f\to f(x)$, $x \in \{x_1,\dots,x_{n+1}\}$, that is
\be
(\Rb^{\{x_1,\hdots,x_{n+1}\}})^*=\span(\delta_{x_1},
\hdots,\delta_{x_{n+1}}). \ee
Moreover, the dual norm $N^*$ is defined by 
\be
(N^*(\mu))^2=\sup\{|\mu(f)|^2:\quad N(f) \leq 1\}, 
\ee 
for all $\mu\in (\Rb^{\{x_1,\hdots,x_{n+1}\}})^*$.
\\
Now, for any function $f\in\Rb^{\{x_1,\hdots, x_{n+1}\}}$ and any sequence of real numbers $(w_1,\hdots,w_n)$, we define the \textbf{linear prediction} of 
$f(x_{n+1})$
\be 
\hat{f}(x_{n+1})=\sum_{i=1}^nw_if(x_i), 
\ee
with the error  
\be
\Ec rr_n(f):=\vert f(x_{n+1})-\sum_{i=1}^nw_if(x_i)\vert, 
\ee
and the worst error in the unit ball w.r.t. the norm $N(.)$ 
\ben
\label{weub}
\Wc err(f):=\sup\{|f(x_{n+1})-\sum_{i=1}^nw_if(x_i)|^2:\quad N(f)\leq
1\}=(N^*(\delta_{x_{n+1}}-\sum_{i=1}^nw_i\delta_{x_i}))^2. 
\een
\par
In the rest oh this paper, we endow the vector space $\Rb^{\{x_1,\hdots, x_{n+1}\}}$ with the scalar inner product
\be
(f,f)=(f,f)_{\K^{-1}}&=&\sum_{i=1}^{n+1}\sum_{j=1}^{n+1}f(x_i)f(x_j)k^{(-1)}(x_i,x_j)\\
&=&\f^\top\K^{-1}\f, 
\ee
with $\f=(f(x_1),\hdots, f(x_{n+1}))^\top$ and  $\K=[k(x_i,x_j):\quad i,j=1,\hdots, n+1]$ is a fixed $(n+1)\times (n+1)$ symmetric positive definite matrix, with $k^{(-1)}(x_i,x_j)$ denotes the $(i,j)$ entry of $\K^{-1}$. The norm defined by $\K$ is given by $N(f)=\|\K^{-1/2}\f\|$, with $\|\cdot\|$ denotes the Euclidean norm.
\subsection{Min-max prediction and kernel interpolation}
\begin{definition}[Min-max prediction]
	A linear prediction $ f^*(x_{n+1}) $ of  $ f(x_{n+1}) $ is called min-max if
	\ben 
	f^*(x_{n+1})=\sum_{i=1}^nw_i^*f(x_i),\label{w*prediction} 
	\een
	where $ \left( w_1^*,...,w_n^*\right)  $ are given by the minimization of the $\Wc err(f)$ \ref{weub}.
\end{definition}
The following result give us the optimal weights associate to the min-max prediction w.r.t. to the norm $\|\K^{-1/2}\cdot\|$.
\begin{prop}
	The the worst error in the unit ball, $\Wc err(f)$,  w.r.t. to the norm $\|\K^{-1/2}\cdot\|$ is equals
	\ben
	\Wc err(f)= \|\delta_{x_{n+1}}-\sum_{i=1}^nw_i\delta_{x_i})\|_{\K^{1/2}}^2 \label{Kundemibasis}
	\een
	where $\|\cdot\|_{\K^{1/2}}$ denotes the dual norm defined by
	the dual scalar inner product \be
	(\delta_{x_i},\delta_{x_j})_{\K}=k(x_i,x_j), \quad i,j=1,\hdots,
	n+1. \ee 
\end{prop}
\begin{proof}
	From the general theory of reproducing kernel Hilbert spaces,see \cite{Berlinet,Scheuerer}, we have 
	\ben
	&&\sup_{\|\K^{-1/2}\f\|\leq 1}\{|f(x_{n+1})-\sum_{i=1}^nw_if(x_i)|^2\}\nonumber\\
	&&=\sup_{\|\K^{-1/2}\f\|\leq 1}\{[\K^{-1/2}\f]^\top[\K^{1/2}(-w_1,\hdots,-w_n,1)^\top(-w_1,\hdots,-w_n,1)\K^{1/2}][\K^{-1/2}\f]\}\nonumber\\
	&&=\mbox{the largest eigenvalue of $[\K^{1/2}(-w_1,\hdots,-w_n,1)^\top(-w_1,\hdots,-w_n,1)\K^{1/2}]$}\nonumber\\
	&&=\|\K^{1/2}(-w_1,\hdots,-w_n,1)^\top\|^2=(-w_1,\hdots,-w_n,1)\K(-w_1,\hdots,-w_n,1)^\top\nonumber\\
	&&=\|\delta_{x_{n+1}}-\sum_{i=1}^nw_i\delta_{x_i})\|_{\K^{1/2}}^2.\nonumber
	\een
\end{proof}
\begin{cor}
	The optimal weights of the min-max linear prediction of	$f(x_{n+1})$ are given by 
	\ben
	\w^*=(w_1^*,\hdots, w_n^*)=[k(x_{n+1},x_1), \hdots,
	k(x_{n+1},x_n)][k(x_i,x_j):\quad i,j=1,\hdots, n]^{-1}.
	\label{optimalweights} 
	\een
	\begin{proof}
		The optimal weights are given by the
		minimization 
		\ben
		\arg\min\{\|\delta_{x_{n+1}}-\sum_{i=1}^nw_i\delta_{x_i}\|^2_{\K^{1/2}}:\quad
		w_1,\hdots, w_n\in\Rb\}, \label{Koptimization} 
		\een 
		which is the solution of the system 
		\ben
		\sum_{j=1}^nw_jk(x_i,x_j)=k(x_{n+1},x_i),\quad i=1,\hdots, n,
		\label{optimalweightsystem} 
		\een
		it follows easily that $\w^*$ is given by \ref{optimalweights}.
	\end{proof}
\end{cor}

\begin{remarks} 
	\begin{itemize}
		\item[1)] The worst case linear prediction error in the ball with the radius $r>0$
		w.r.t. to the norm $\|\K^{-1/2}\cdot\|$ is equal to
		\be
		&&\sup_{\|\K^{-1/2}\f\|\leq r}\{|f(x_{n+1})-\sum_{i=1}^nw_if(x_i)|^2\}\\
		&&=r^2(-w_1,\hdots,-w_n,1)\K(-w_1,\hdots,-w_n,1)^\top, 
		\ee
		as a result the optimal weights (\ref{optimalweights}) do not depend on the radius of the ball. 
		
		\item[2)] The prediction using the spline interpolating w.r.t. the norm $\|\K^{-1/2}\cdot\|$ (see, e.g., \cite{Berlinet}) defined by the minimizer : 
		\be
		\label{spint}
		S(f)=\arg\min\{\|\K^{-1/2}\f\|:\quad
		f(x_1),\hdots,f(x_n)\,\mbox{are fixed}\},
		\ee
		coincide with the prediction (\ref{w*prediction}).
		\item[3)] The min-max prediction (\ref{w*prediction}) is equal to
		\ben 
		f^*(x_{n+1})=\sum_{j=1}^n\alpha_j^* k(x_{n+1},x_j),
		\label{kernelinterpolation}
		\een 
		where 
		$ \alpha_j^*,j=1,...,n $ is the solution of the system 
		\ben 
		\sum_{j=1}^n\alpha_j^*\k(x_i,x_j)=f(x_i),\quad i=1,\hdots, n, 
		\label{alphainterpolation}
		\een 
	\end{itemize}
\end{remarks}
\par
Now, we turn to the interpolation of the function $f$ at the set $\{x_1,\hdots,x_{n}\}$ using $span(\k_1,\hdots,\k_n)$ where
$\k_j$ denotes the $j$-th column of the matrices $\K$. Then the interpolation of the function $f$  equals 
\be 
I(f)=\sum_{j=1}^n\alpha_j^*\k_j
\ee 
with the weights $\alpha^*$ are given by (\ref{alphainterpolation}). 
The following Proposition gives the error of interpolation.
\begin{prop}[Interpolation error ]\label{interpolationerror} 
	The error of interpolation, $ \Ic\Ec rr $, is given by 
	\ben
	\Ic\Ec rr(f):=f(x_{n+1})-f^*(x_{n+1})=[\k_{n+1}^{(-1)}\f][k(x_{n+1},x_{n+1})-\sum_{i=1}^nw_i^*k(x_i,x_{n+1})].
	\label{realerror}
	\een
	\begin{proof}
		First, observe that we can write the coordinates of $\f$ in the basis $\K$ as
		\be 
		\f=\sum_{j=1}^{n+1}[\k_j^{(-1)}\f]\k_j.
		\ee 
		with $\k_j^{(-1)}$ denotes the $j$-th row of $\K^{-1}$. Therefore
		\be
		I(f)&=&\sum_{j=1}^{n+1}[\k_j^{(-1)}\f]I(\k_j)\\
		&=&\sum_{j=1}^{n}[\k_j^{(-1)}\f]\k_j+[\k_{n+1}^{(-1)}\f]I(\k_{n+1}), \ee
		because the interpolation of $\k_{j}$ is exact for $j=1,\hdots, n$.
		Thus, 
		\be
		f-I(f)&=&[\K_{n+1}^{-1}\f](\k_{n+1}-I(\k_{n+1}))\\&=&[\k_{n+1}^{(-1)}\f](0,\hdots,0,[k(x_{n+1},x_{n+1})-\sum_{i=1}^nw_i^*k(x_i,x_{n+1})])^\top,
		\ee which completes the proof.
	\end{proof}
\end{prop}

\subsection{Min-max linear prediction with constraint}
In this section, we consider the optimization (\ref{Koptimization}) under the
constraint \ben \sum_{i=1}^nw_i\p_k(x_i)=\p_k(x_{n+1}),\quad
k=1,\hdots, q, \label{wconstraint} \een where $\p_1,\hdots,\p_q\in
\Rb^{\{x_1,\hdots, x_{n+1}\}}$ are given.
\\
Solve the minimization (\ref{Koptimization}) under the constraint  (\ref{wconstraint})
is equivalent to solve the system
%
\ben
\begin{cases} \sum_{j=1}^nk(x_i,x_j)w_j+\sum_{k=1}^q\lambda_k\p_k(x_i)=k(x_i,x_{n+1}),\quad i=1,\hdots, n, \\ \sum_{j=1}^nw_j\p_k(x_j)=\p_k(x_{n+1}),\quad k=1,\hdots, q, \label{wlambda}\end{cases}
\een
where $\lambda_1$, $\hdots$, $\lambda_q$ are the Lagrange multiplier.
The solution is unique if the homogeneous system
\be
&&\sum_{j=1}^nk(x_i,x_j)w_j+\sum_{k=1}^q\lambda_kp_k(x_i)=0,\quad i=1,\hdots, n,\\
&&\sum_{j=1}^nw_jp_k(x_j)=0,\quad k=1,\hdots, q, \ee has a unique
solution $w_1=\hdots=w_n=0$, $\lambda_1=\hdots=\lambda_q=0$. This is
equivalent to say that the columns
$(p_k(x_1),\hdots,p_k(x_n))^\top$, $k=1,\hdots, q$, are linearly
independent and that $\K$ is conditionally positive w.r.t. $\p_1$,
$\hdots$, $\p_q$, i.e. the system \be
\sum_{j=1}^n\sum_{j=1}^nk(x_i,x_j)w_jw_i=0,\quad
\sum_{i=1}^nw_ip_k(x_i)=0,\quad k=1,\hdots,q, \ee has a unique
solution $w_1=\hdots=w_n=0$. Observe that this is true if $\K$ is definite positive, but it is not necessary.
\\
Let $w_1^*$, $\hdots$, $w_n^*$, $\lambda_1^*$, $\hdots$,
$\lambda_q^*$ be the solution of the system (\ref{wlambda}). Then
the optimal prediction under the constraint (\ref{wconstraint}) is
\ben f^*(x_{n+1})=\sum_{i=1}^nw_i^*f(x_i).
\label{optimalpredictionunderconstraint} \een
\subsection*{Constraint's parametrization}
Now, let $\z^{(1)}=(z_1^{(1)},\hdots,z_n^{(1)})^\top$ be a particular solution of the system (\ref{wconstraint}) and $\z_1=(z_{11},\hdots,z_{n1})^\top$, $\hdots$, $\z_{n-q}=(z_{1n-q},\hdots,z_{nn-q})^\top$,
$n-q$ independent solutions of the corresponding homogeneous system. Then the general solution
of the system (\ref{wconstraint}) has the form
\be
\w=\z^{(1)}+\sum_{l=1}^{n-q}\tilde{w}_l\z_l.
\ee
Let us consider the basis
\ben
&&\B=[\p_1,\hdots,\p_q,\sum_{i=1}^nz_{i1}\k_i,\hdots,\sum_{i=1}^nz_{in-q}\k_i,\b_{n+1}]\label{pzkbasis}\\
&&=:[\b_1,\hdots,\b_{n+1}]\nonumber
\een
such that the expansion of $\f$ in the basis $\B$ is given by  
\be 
\f=\sum_{l=1}^q\theta_l\f\p_l+\sum_{l=1}^{n-q}\theta_{q+l}\f\b_{q+l}+\theta_{n+1}\f\b_{n+1},
\ee 
with $\theta_{n+1}=(-(\z^{(1)})^\top,1)$. 

The unknows are the rows $\theta_1$, $\hdots$, $\theta_{n}$, and the column 
$\b_{n+1}$. They are solution of the system
\ben 
&&\delta(i=j)=\sum_{l=1}^qp_l(x_i)\theta_{lj}+
\sum_{l=1}^{n-q}b_{q+l}(x_i)\theta_{(q+l)j}+b_{n+1}(x_i)\theta_{(n+1)j}\nonumber\\ 
&&i,j=1,\hdots, n+1.\label{systemforconstraintbasis}
\een 
\\
The interpolation of any function $g$ at the set
$\{x_1,\hdots,x_{n}\}$ using $span(\b_1,\hdots,\b_n)$ is given by
the map \be I(g)=\sum_{j=1}^n\beta_j\b_j, \ee with $\beta$ is the
solution of the system \be \sum_{j=1}^n\beta_j\b_j(x_i)=g(x_i),\quad
i=1,\hdots, n. \ee Using similar arguments as in proposition
\ref{interpolationerror}, we can deduce    the following result.
\begin{prop} The value
	\be
	\sum_{j=1}^nz_j^{(1)}k(x_j,x_{n+1})+I(f-\sum_{j=1}^nz_j^{(1)}\k_j)(x_{n+1})
	\ee coincides with $f^*(x_{n+1})$ given by
	(\ref{optimalpredictionunderconstraint}). In addition, the error is
	equal to \be
	&&f(x_{n+1})-f^*(x_{n+1})=f(x_{n+1})-\sum_{j=1}^nz_j^{(1)}k(x_j,x_{n+1})-
	I(f-\sum_{j=1}^nz_j^{(1)}\k_j)(x_{n+1})\\
	&&=[\b_{n+1}^{(-1)}(f-\sum_{j=1}^nz_j^{(1)}\k_j)][b_{n+1}(x_{n+1})-I(\b_{n+1})(x_{n+1})].
	\ee
\end{prop} 

If $\K$ is invertible and $\p_l=\k_l$ with $l=1,\hdots, n$, then $\theta_l=\k_l^{(-1)}$
for $l=1,\hdots,n$, and the basis (\ref{systemforconstraintbasis}) 
is given by $\b_l=\k_l$ with $l=1,\hdots, n$, and 
$\b_{n+1}=\frac{\k_{n+1}}{k(x_{n+1},x_{n+1})-I(\k_{n+1})(x_{n+1})}$.  
\subsection*{Constraint's effect on the kernel}
From the notations above the general solution of the system
(\ref{wconstraint}) has the form \be
\w=\z^{(1)}+\sum_{l=1}^{n-q}\tilde{w}_l\z_l. \ee 
As a consequence the quadratic
form \be \|\delta_{x_{n+1}}-\sum_{i=1}^nw_i\delta_{x_i}\|_{\K}^2=
\|\mu_{n+1-q}-\sum_{l=1}^{n-q}\tilde{w}_l\mu_l\|_{\tilde{\K}}^2, \ee
with $\mu_1=\sum_{i=1}^nz_{i1}\delta_{x_i}$, $\hdots$,
$\mu_{n-q}=\sum_{i=1}^nz_{in-q}\delta_{x_i}$,
$\mu_{n+1-q}=\delta_{x_{n+1}}-\sum_{i=1}^nz_i^{(1)}\delta_{x_i}$,
and the entries of the $(n+1-q)\times (n+1-q)$ kernel $\tilde{\K}$
are given by \be \tilde{k}(l_1,l_2)=(\mu_{l_1},\mu_{l_2})_{\K},\quad
l_1,l_2=1,\hdots, n+1-q. \ee Observe that $\tilde{\K}$ is positive
definite if and only if the columns
$(p_k(x_1),\hdots,p_k(x_n))^\top$, $k=1,\hdots, q$, are linearly
independent and   $\K$ is conditionally positive w.r.t. $\p_1$,
$\hdots$, $\p_q$.
\\
It follows that
\be
&&\sup\{|f(x_{n+1})-\sum_{i=1}^nw_if(x_i)|^2:\quad \tilde{f}^\top\tilde{\K}^{-1}\tilde{f}\leq 1\}\\
&&=\|\mu_{n+1-q}-\sum_{l=1}^{n-q}\tilde{w}_l\mu_l\|_{\tilde{\K}}^2,
\ee
where $\tilde{f}=(\tilde{f}(1),\hdots, \tilde{f}(n+1-q))^\top\in \Rb^{n+1-q}$ are defined by
\be
&&\tilde{f}(1)=\sum_{i=1}^nz_{i1}f(x_i),\hdots,\tilde{f}(n-q)=\sum_{i=1}^nz_{in-q}f(x_i),\\
&&\tilde{f}(n+1-q)=f(x_{n+1})-\sum_{i=1}^nz_i^{(1)}f(x_i).
\ee
The map $f\in \Rb^{\{x_1,\hdots,x_{n+1}\}}\to \tilde{f}^\top\tilde{\K}^{-1}\tilde{f}$
is a semi kernel having the null space spanned by $\p_1$, $\hdots$, $\p_q$.
\\
That being the case, the optimal weights $\tilde{w}^*$ are given by \be
\tilde{w}^*=\arg\min\{\|\mu_{n+1-q}-\sum_{l=1}^{n-q}\tilde{w}_l\mu_l\|_{\tilde{\K}}^2:\quad
\tilde{w}_1,\hdots, \tilde{w}_{n-q}\in\Rb\}, \ee and then predict
$f(x_{n+1})$ is equal to \be
\sum_{i=1}^nz_i^{(1)}f(x_i)+\sum_{l=1}^{n-q}\tilde{w}_l^*(\sum_{i=1}^nz_{il}f(x_i)).
\ee The latter predictor coincides with
(\ref{optimalpredictionunderconstraint}). Moreover, the spline \be
S(\tilde{f})=\arg\min_{\tilde{f}(n+1-q)}\{\tilde{f}^\top\tilde{K}^{-1}\tilde{f}:\quad
\tilde{f}(1),\hdots,
\tilde{f}(n-q)\,\mbox{are fixed}\} \ee is such that
\be S(\tilde{f})(n+1-q)=f^*(x_{n+1})-\sum_{i=1}^nz_i^{(1)}f(x_i) \ee
with $f^*(x_{n+1})$ is the optimal prediction under the constraint
(\ref{optimalpredictionunderconstraint}).
\\
From the expansion of $\f$
in the basis $\B=[\b_1,\hdots,\b_{n+1}]$ (\ref{pzkbasis}), we can conclude the following result.

\begin{prop}
	If the weights $\w$ satisfy the constraint (\ref{wconstraint}), then
	\be
	|f(x_{n+1})-\sum_{i=1}^nw_if(x_i)|^2=|\sum_{l=1}^{n+1-q}\tilde{f}_{l}\{b_{q+l}(x_{n+1})-\sum_{i=1}^nw_ib_{q+l}(x_i)\}|^2.
	\ee
	It follows that
	\be
	&&\sup\{|f(x_{n+1})-\sum_{i=1}^nw_if(x_i)|^2:\quad \tilde{f}^\top\tilde{\K}^{-1}\tilde{f}\leq 1\}\\
	&&=(b_{q+1}(x_{n+1})-\sum_{i=1}^nw_ib_{q+1}(x_i),\hdots,b_{n+1}(x_{n+1})-\sum_{i=1}^nw_ib_{n+1-q}(x_i))\tilde{\K}^{-1}\\
	&&(b_{1+q}(x_{n+1})-\sum_{i=1}^nw_ib_{1+q}(x_i),\hdots,b_{n+1}(x_{n+1})-\sum_{i=1}^nw_ib_{n+1}(x_i))^\top\\
	&&=(-w_1,\hdots,-w_n,1)\R \tilde{K}\R^\top(-w_1,\hdots,-w_n,1)^\top\\
	&&=(-\tilde{w}_1,\hdots,-\tilde{w}_{n-q},1)\tilde{K}(-\tilde{w}_1,\hdots,-\tilde{w}_{n-q},1)^\top\\
	&&=\|\mu_{n+1-q}-\sum_{l=1}^{n-q}\tilde{w}_l\mu_l\|_{\tilde{\K}}^2,
	\ee with the $(n+1)\times (n+1-q)$ matrix \be
	\R=[\b_{q+1},\hdots,\b_{n+1}]. \ee
\end{prop}
\subsection{Semi-kernel and constraint}
Now, conversely we consider a semi-kernel $Q$ on
$\Rb^{\{x_1,\hdots,x_{n+1}\}}$ with the null space spanned by $q$
functions $p_1$, $\hdots$, $p_q$ and let \be
S(f)=\arg\min\{Q(f,f):\quad f(x_1),\hdots,f(x_n)\mbox{ are fixed}\}
\ee be the spline defined by the semi-norm $Q$, and
$$S(f)(x_{n+1})=\sum_{i=1}^nw_i^*f(x_i).$$ We consider a basis
$\B=[\b_1,\hdots,\b_{n+1}]$ such that $\b_l=\p_l$ with
$l=1,\hdots,q$ and let
$(\theta_1,\hdots,\theta_q,u_1,$ $ \hdots, u_{n+1-q})$ be the coordinates
of $\f$, i.e. \be
\f=\sum_{l=1}^q\theta_l\p_l+\sum_{l=1}^{n+1-q}u_l\b_{q+l}. \ee It
follows that \be
Q(f,f)=\sum_{l_1=1}^{n+1-q}\sum_{l_2=1}^{n+1-q}u_{l_1}u_{l_2}
Q(b_{q+l_1},b_{q+l_2})=\|Q^{1/2}\u\|^2, \ee and the kernel
$Q=:[Q_{l_1,l_2}:\quad l_1,l_2=1,\hdots, n+1-q]$ is invertible. If
the weights $\w$ satisfy the constraint (\ref{wconstraint}), then
\be
|f(x_{n+1})-\sum_{i=1}^nw_if(x_i)|^2=|\sum_{l=1}^{n+1-q}u_l\{b_l(x_{n+1})-\sum_{i=1}^nw_ib_l(x_i)\}|^2,
\ee therefore, \be
&&\sup\{|f(x_{n+1})-\sum_{i=1}^nw_if(x_i)|^2:\quad Q(f,f)\leq 1\}\\
&&=\sup\{|f(x_{n+1})-\sum_{i=1}^nw_if(x_i)|^2:\quad \u^\top Q\u\leq 1\}\\
&&=(b_{q+1}(x_{n+1})-\sum_{i=1}^nw_ib_{q+1}(x_i),\hdots,b_{n+1}(x_{n+1})-\sum_{i=1}^nw_ib_{n+1-q}(x_i))Q^{-1}\\
&&(b_{1+q}(x_{n+1})-\sum_{i=1}^nw_ib_{1+q}(x_i),\hdots,b_{n+1}(x_{n+1})-\sum_{i=1}^nw_ib_{n+1}(x_i))^\top\\
&&=(-w_1,\hdots,-w_n,1)\R Q^{-1}\R^\top(-w_1,\hdots,-w_n,1)^\top,
\ee where the $(n+1)\times (n+1-q)$ matrix \be
\R=[\b_{q+1},\hdots,\b_{n+1}]. \ee
\section{Stochastic approach}
The statistical counterpart to the kernel interpolation is known as
kriging (see e.g. \cite{Scheuerer}). It is
based on the modeling assumption that
$(f(x_1),\hdots, f(x_n),f(x_{n+1}))$ is a realization of random
vector $Y_{x_1}$, $\hdots$, $Y_{x_{n+1}}$ over the same probability space $(\Omega,\mathcal{F},\P)$.
To predict $Y_{x_{n+1}}$ known $Y_{x_1}$, $\hdots$, $Y_{x_n}$  we need the 
mean and the covariance matrix of the random vector $(Y_{x_1},\hdots, Y_{x_{n+1}})$.
\\
We assume that the mean $(m(x_1),\hdots,m(x_{n+1}))$ (also called the trend) and
the covariance function
\be
k(x_i,x_j)=cov(Y_{x_i},Y_{x_j})
\ee
of the random vector $(Y_{x_1},\hdots, Y_{x_{n+1}})$ exist.
\\
If $Y_{x_1}$, $\hdots$, $Y_{x_n}$ are assumed to be known, then the
best linear unbiased predictor (BLUP) of $Y_{x_{n+1}}$
is given by \be
\hat{Y}_{x_{n+1}}=\sum_{i=1}^nw_i^*Y_{x_i}, \ee where the weights
$w_i^*$ are the solution of the following optimization problem \ben
\min\{var(Y_{x_{n+1}}-\sum_{i=1}^nw_iY_{x_i}):\quad w_1,\hdots,
w_n\in\Rb, \sum_{i=1}^nw_im(x_i)=m(x_{n+1})\}. \label{BLUP} \een If
the mean function $m$ is modeled as \be
m(x_i)=\sum_{k=1}^q\beta_kp_k(x_i):\quad i=1,\hdots, n+1, \ee and if
we consider the weights such that \be
\sum_{i=1}^nw_ip_l(x_i)=p_l(x_{n+1}),\quad l=1,\hdots, q, \ee then
the optimal predictor \be \hat{f}(x_{n+1})=\sum_{i=1}^nw_i^*f(x_i)
\ee of $f(x_{n+1})$ in stochastic sense coincides with the
interpolation (\ref{optimalpredictionunderconstraint}).

\section{Three kernel selection criteria}
Kernel interpolation and prediction approaches are based on the
knowledge of a symmetric positive definite matrix $\K$ and the trend
$\p_1$, $\hdots$, $\p_q$. To apply  kernel interpolation it amounts
to the assumption that one knows the degree of smoothness of the
function $f$. In the context of partial differential equations, the
function $f$ belongs to some Sobolev space. In stochastic approach
the covariance matrix and the trend are chosen using the maximum
likelihood method or the Bayesian method.
\\
Here we propose three natural criteria to  compare two kernels $\K^{(1)}$ and $\K^{(2)}$.
Known $f(x_1)$, $\hdots$, $f(x_r)$,
we predict $f(x_{r+1})$ using the kernel $[k^{(l)}(i,j):\quad i,j=1,\hdots,r]$, and we obtain
the predictor $\hat{f}^{(l)}(x_{r+1})$, with $l=1,2$, and $r=2$, $\hdots$, $n-1$.
We propose the following  three criteria to measure the performance
of the Kernel $\K^{(l)}$:

1) $MSPE(l)=:\frac{\sum_{j=1}^{n-1}|f(x_{j+1})-\hat{f}^{(l)}(x_{j+1})|^2}{n-1}$.
We say that $\K^{(1)}$ is better than $\K^{(2)}$
w.r.t. the MSPE criterion
if
\be
MSPE(1)<MSPE(2).
\ee

2) $MAXPE(l)=:\max\{|f(x_{j+1})-\hat{f}^{(l)}(x_{j+1})|:\quad j=1,\hdots,n-1\}$.
We say that $\K^{(1)}$ is better than $\K^{(2)}$ w.r.t. the MAXPE criterion
if
\be
MAXPE(1)<MAXPE(2).
\ee

3) We say that $\K^{(1)}$ is statistically better than $\K^{(2)}$ if
\be
\frac{\sum_{j=1}^{n-1}{\bf 1}_{[|f(x_{j+1})-\hat{f}^{(1)}(x_{j+1})| <
		|f(x_{j+1})-\hat{f}^{2)}(x_{j+1})|]}}{n-1} >1/2.
\ee
These criteria were also used in \cite{DEM}.

\section{Application}
In the climate change problem we are interested in the mean temperature 
$f(t)$ at the time $t$. 
The data are the years taken into account $t_1< \hdots< t_{n+1}$ and the mean 
temperature $f(t_1)$, $\hdots$, $f(t_n)$, and we are interested in the prediction of $f(t_{n+1})$. 
We recall that 
\be
\arg\min\{\int_{t_1}^{t_{n+1}}|g''(t)|^2dt:\quad g(t_1)=f(t_1),\hdots, g(t_{n+1})=f(t_{n+1})\quad\mbox{are fixed}\}
\ee 
is the natural $C^2$ cubic spline $s$ which interpolates the points $(t_i,f(t_i)), i=1, \ldots, n+1$.
See \cite{schoenberg,reinsch}. We assume that $f(t_1)$, $\hdots$, $f(t_{n+1})$ are the values of 
a natural $C^2$ cubic spline. We are going to predict 
$f(t_{n+1})$ using three kernels, and we need some notations. 

\subsection{Kernel and semikernels using cubic splines}
Let $S = S_3(t_1, \ldots, t_{n+1})$ be the set of $C^2$ cubic
splines having the knots $t_1 <\cdots < t_{n+1}$. Note that every
element $s\in S$ is a $C^2$ map on $[t_1, t_{n+1}]$ and is a
polynomial of degree three on each interval $[t_i,t_{i+1})$ for
$i=1$,\ldots, $n$.

More precisely, let
\be
&&p_1 =s(t_1), \ldots, p_{n+1} =s(t_{n+1}),\quad q_1 =s'(t_1), \ldots, q_{n+1} =s'(t_{n+1}),\\
&&u_1 =s''(t_1), \ldots, u_{n+1} =s''(t_{n+1}),\quad v_1 =
s'''(t_1+), \ldots, v_{n}=s'''(t_{n}+) \ee be respectively the
values of $s$ and its derivatives up to order three at the knots. We
have for every $i=1,\ldots, n$, \be
s(t)=p_i+q_i(t-t_i)+(t-t_i)^2u_i/2+(t-t_i)^3v_i/6,\quad t\in [t_i,
t_{i+1}). \ee

The following constraint for $h_i=t_{i+1}-t_i$, $i=1, \ldots, n$
ensures the hypothesis that $s$ is $C^2$: \ben
&&p_i+q_ih_i+u_ih_i^2/2+v_ih_i^3/6=p_{i+1},\quad \label{c0}\\
&&q_i+u_ih_i+v_ih_i^2/2=q_{i+1},\quad \label{c1}\\
&&v_i=s^{(3)}(t_i)=(u_{i+1}-u_i)/h_i.\quad \label{c2}
\een
\\
It is well known (see \cite{deBoor}) that $S$ has the dimension $n+3$, and the set 
of natural spline $S_{nat}$ has the dimension $n+1$. 
Hence an element $s\in S$ (respectively $s\in S_{nat}$) is completely defined by 
$n+3$ (respectively $n+1$ parameters) independent parameters.   

Now we need to parametrize the set $S$ in order to define properly an element  
$s\in S$. A parametrization of $S$ is a one-to-one linear map 
\be 
\Theta: s\in S\to \theta\in \Rb^{n+3}.
\ee 
Defining a parametrization $\Theta$ 
is equivalent to the existence of the basis 
$\B=(\b_1, \ldots, \b_{n+3})$ of $S$ such that, for all $s\in S$,
\be 
s=\sum_{i=1}^{n+3}\theta_i \b_i =\B\theta. 
\ee 
The parametrization   
$\Theta_{002}=(p_1, p_2, u_1, \ldots, u_{n+1})$ defines the basis 
$\B_{002}=(\b_1^{002}, \ldots, \b_{n+3}^{002})$.  
The subscript notation 002 is justified by the fact that 
\be &&p_1 =s(t_1) =s^{(0)}(t_1), p_2 =s(t_2) =s^{(0)}(t_2),\\ &&u_1 =s''(t_1) =s^{(2)}(t_1), \ldots, u_{n+1} =s''(t_{n+1}) =s^{(2)}(t_{n+1}). 
\ee 
See \cite{DP1,DP2,DP3} for more details. 

It follows for $s\in S_{nat}$ that   
\be 
s=p_1\b_1^{002}+p_2\b_2^{002}+\sum_{i=2}^{n}u_i\b_{2+i}^{002},
\ee 
and then $\s=(s(t_1),\hdots,s(t_{n+1})^\top$ is given by 
\be 
\s=[\b_1^{002}(\t),\b_2^{002}(\t)](p_1,p_2)^\top+\sum_{i=2}^{n}\R(u_2,\hdots,u_{n})^\top,
\ee 
Here the column $\b_i^{002}(\t)=(b_i^{002}(t_1),\hdots,b_i^{002}(t_{n+1}))^\top$, with $i=4,\hdots,n+2$, and the $n+1\times n+1$ matrix  
\be 
\R=[\b_{4}^{002}(\t),\hdots,\b_{n+2}^{002}(\t)]. 
\ee 
Observe that $span(\b_1^{002},\b_2^{002})=span({\bf 1},{\bf t})$ with 
the column ${\bf 1}=(1,\hdots,1)^\top$, ${\bf t}=(t_1,\hdots,t_{n+1})^\top$.

We can show that 
\ben 
&&\int_{t_1}^{t_{n+1}}|s''(t)|^2dt=\nonumber\\ 
&&\sum_{i=1}^n\int_{t_i}^{t_{i+1}}|u_i+t(u_{i+1}-u_i)/h_i|^2dt\nonumber\\
&=&\sum_{i=1}^n(u_i^2+u_iu_{i+1}+u_{i+1}^2)h_i/3\nonumber\\
&=&(u_2,\hdots,u_n)\Q(u_2,\hdots,u_n)^\top,\label{Q}   
\een
with $\Q$ is a known $n-1\times n-1$ invertible matrix see \cite{DP1}. 
We also recall that 
\be 
(u_2,\hdots, u_n)^\top=\U(p_1,\hdots,p_{n+1})^\top,
\ee 
with $\U$ is a known $n-1\times n+1$ matrix see \cite{DP1}. Therefore  
\ben 
&&(u_2,\hdots, u_n)\Q(u_2,\hdots,u_n)^\top=(p_1,\hdots,p_{n+1})\U^\top\Q\U(p_1,\hdots,p_{n+1})^\top\label{Q}\\
&&=:(p_1,\hdots,p_{n+1})\P(p_1,\hdots,p_{n+1})^\top\label{P}\\
\een 

Now we propose the following predictors for $f(t_{n+1})$. 

0) We assume that $\s$ is Gaussian centred with the covariance matrix $\K^{(0)}=(\Q^{(0)})^{-1}$ with $\Q^{(0)}$ is defined by  
\be 
\int_{t_1}^{t_{n+1}}|s(t)|^2dt=\s^\top\Q^{(0)}\s.
\ee

1) We consider the spline 
\ben 
S(f)=\arg\min\{(f(t_1),\hdots,f(t_{n+1}))\P(f(t_1),\hdots,f(t_{n+1}))^\top:\quad f(t_1),\hdots, f(t_n)
\quad \mbox{are fixed}\},\label{Pspline}
\een 
defined by the kernel $\P$ (\ref{P})
and the predictor $f^*(t_{n+1})=S(f)(t_{n+1})$ of $f(t_{n+1})$. We assume that $\s$ is Gaussian
with the mean $p_1\b_1^{002}(\t)+p_2\b_2^{002}(\t)=\beta_1{\bf 1}+\beta_1{\bf t}$ and the covariance matrix $\K^{(1)}=\R\Q^{-1}\R^\top$ with the kernel $\Q$ is given by (\ref{Q}). The predictor $\hat{f}^{(1)}(t_{n+1})$ of $f(t_{n+1})$ (\ref{optimalpredictionunderconstraint}) using the kernel $\K^{(1)}$ coincides with 
$S(f)(t_{n+1})$. 

2) We assume that $\s$ is Gaussian
with the mean $p_1\b_1^{002}(\t)+p_2\b_2^{002}(\t)=\beta_1{\bf 1}+\beta_1{\bf t}$ and the covariance matrix $\K^{(2)}=\R\Q\R^\top$.

Let $\hat{f}^{(i)}(t_{n+1})$ be the predictor of $f(t_{n+1})$ (\ref{optimalpredictionunderconstraint}) using the kernel $\K^{(i)}$ with $i=0,1,2$. Using real data, we are going to compare these three predictors.

\subsection{Real data Application}
As application in the climate change area we are interested in the annual mean temperature observed in France and Morocco from 
1901 to 2015, the data are presented in Figure \ref{data}.
We illustrate the importance of the kernel choice by considering the kernels 
$\K^{(0)}$, $\K^{(1)}$, $\K^{(2)}$. The three kernel selection criteria 
are presented in Table \ref{err}. The mean annual temperature of the year 2015 and 2016 (i.e. $\hat{f}^{(i)}(t_{n})$ and $\hat{f}^{(i)}(t_{n+1})$, $n=114$) are given in Tables (\ref{pr2015}, \ref{pr2016}), as for Figure \ref{data s} it shows the splines of the predictors $\hat{f}^{(0)}$, $\hat{f}^{(1)}$, $\hat{f}^{(2)} $ and the true temperature. The $w_1^*$, $\hdots$, $w_n^*$ of (\ref{optimalpredictionunderconstraint}) for the kernels $\K^{(0)}$, $\K^{(1)}$, $\K^{(2)}$ are presented in Figure \ref{w}.

\begin{figure}[H]
	\caption{Mean annual temperatures in France and Morocco from 1901 to 2015.}
	\label{data}
	\begin{center}
		{\includegraphics[ width=11cm]{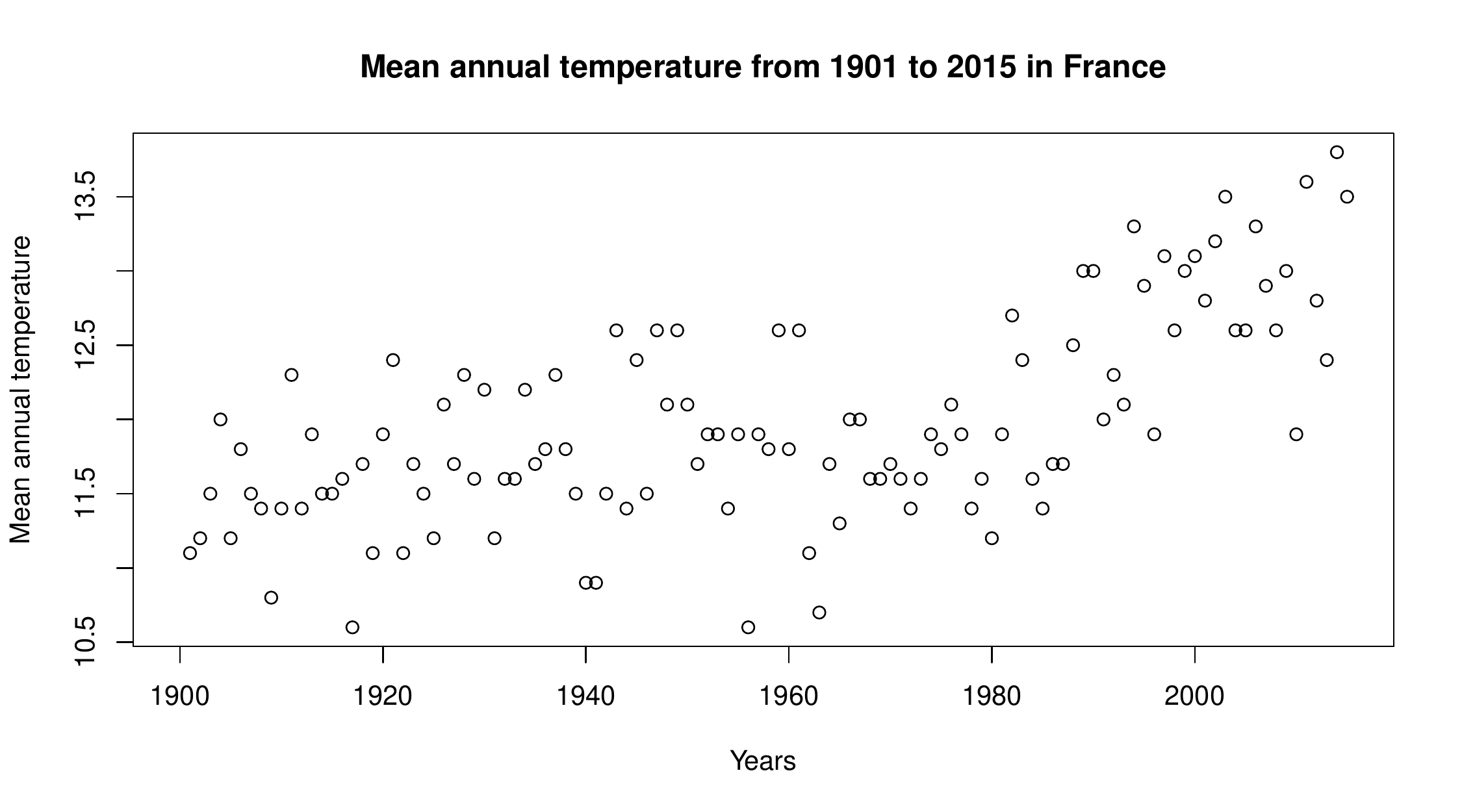}}
		{\includegraphics[ width=11cm]{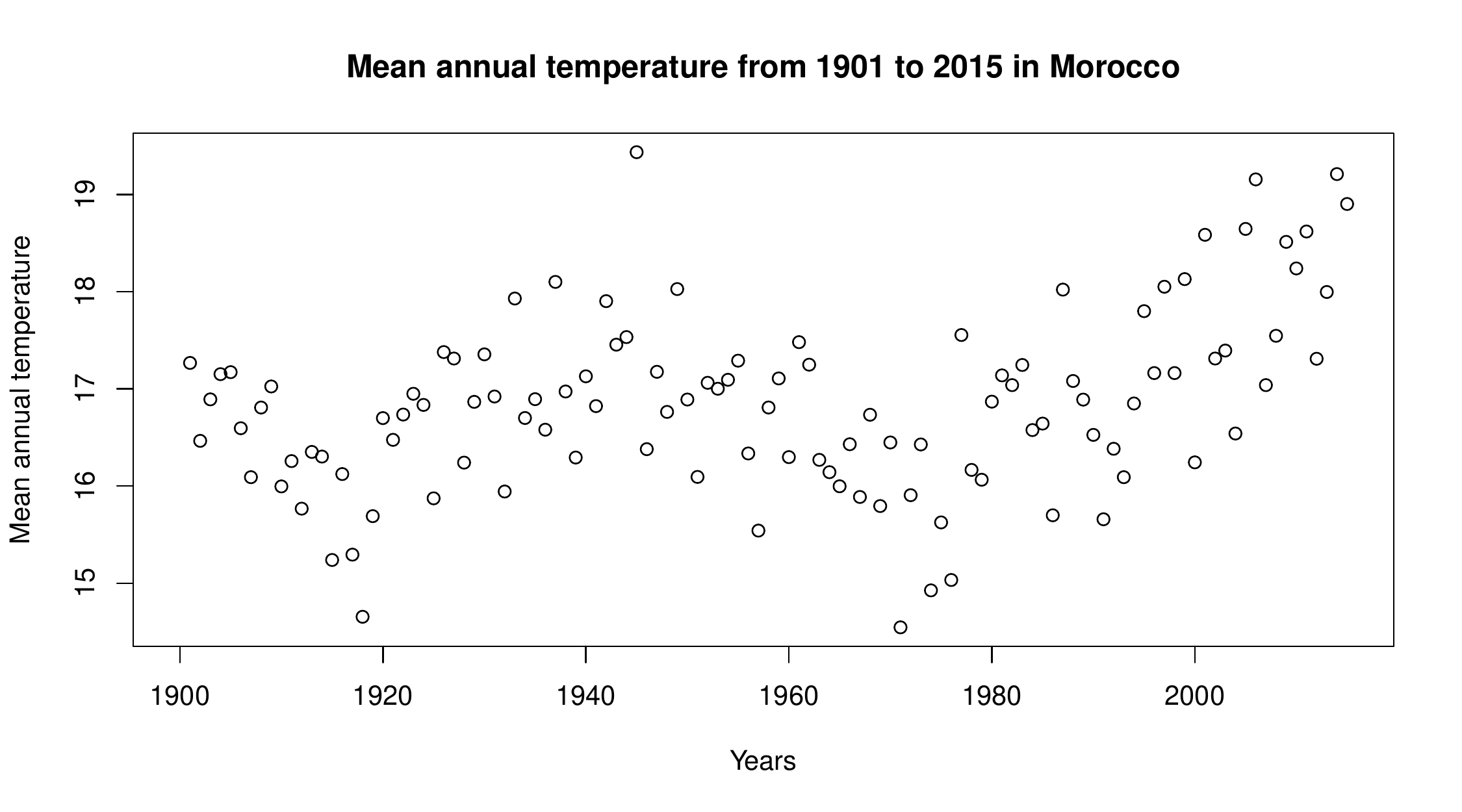}}
	\end{center}
\end{figure}

\begin{table}[H]
	\centering
	\caption{ The three kernel selection criteria for the kernels $\K^{(0)}$,$\K^{(1)}$, $\K^{(2)}$  using Morroco and France data.}
	\label{err}
	\begin{tabular}{|l|l|l|l|l|l|l|}
		\cline{1-7}
		Country &\multicolumn{3}{c|}{\textbf{France}}&\multicolumn{3}{c|}{\textbf{Morocco}}  \\ \cline{1-7}
		Kernel& $\K^{(0)}$ & $\K^{(1)}$& $\K^{(2)}$ & $\K^{(0)}$&$\K^{(1)}$& $\K^{(2)}$ \\ \cline{1-7}
		$MSPE$ &0.3301302&  2.090779 &5.21788 & 0.7727975  & 5.110724& 11.70042\\ \cline{1-7}
		$MAXPE$ &  1.3961  & 3.344106&5.312125 & 2.251341  & 6.171007& 9.438383\\ \cline{1-7}
		Statistically & \multicolumn{3}{c|}{$\K^{(0)}$ with 0.8198198 for $\K^{(1)}$ } & \multicolumn{3}{c|}{$\K^{(0)}$ with 0.8288288 for $\K^{(1)}$ } \\ 
		& \multicolumn{3}{c|}{and 0.8288288 for $\K^{(2)}$} & \multicolumn{3}{c|}{and 0.8468468 for $\K^{(2)}$} \\ \cline{1-7}
	\end{tabular}
	
\end{table}

\begin{figure}[H]
	\caption{The splines of the predictors $\hat{f}^{(0)}$, $\hat{f}^{(1)}$, $\hat{f}^{(2)} $ and the true temperature.}
	\label{data s}
	\begin{center}
		{\includegraphics[ width=12cm]{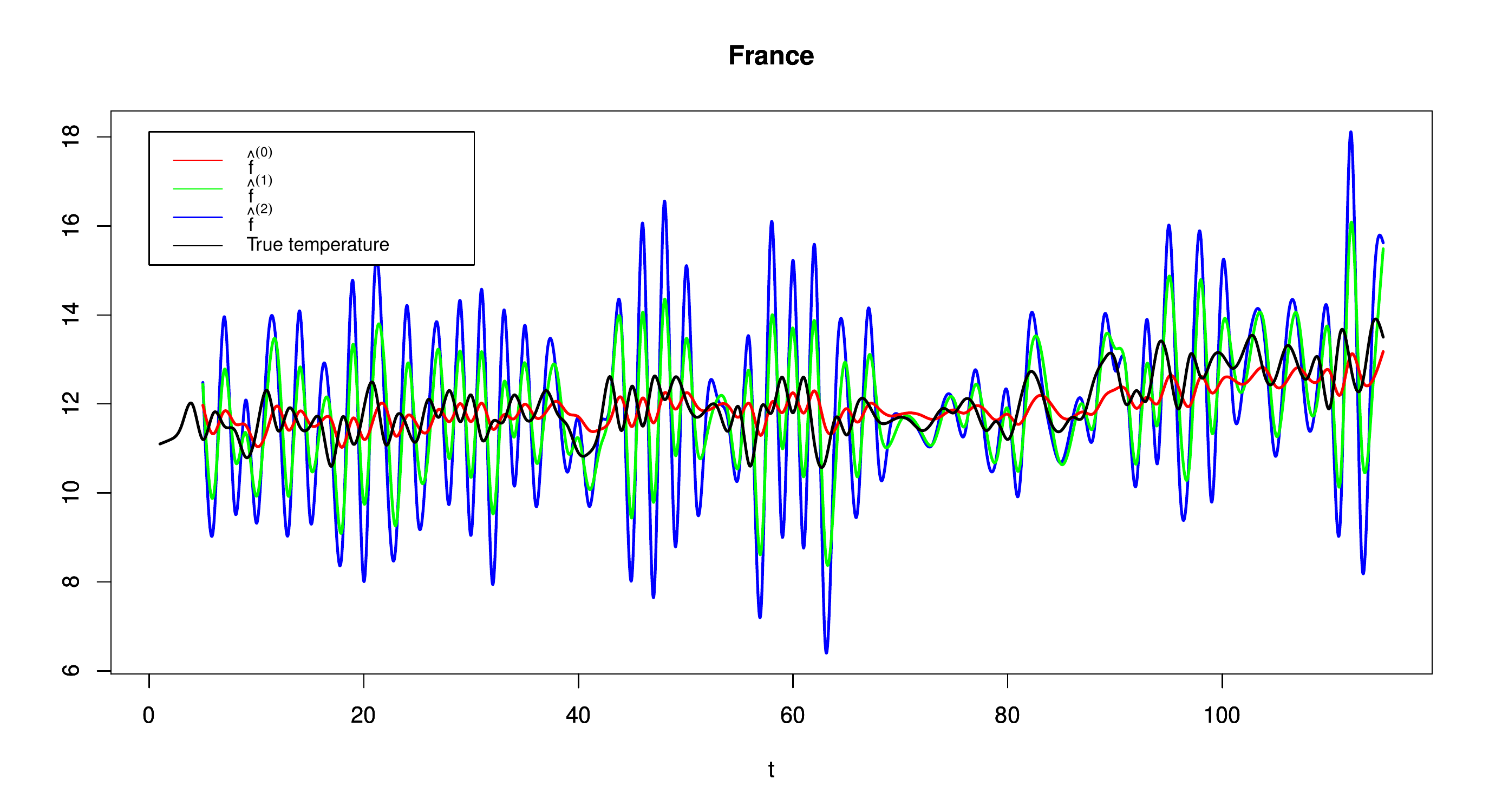}}
		{\includegraphics[ width=12cm]{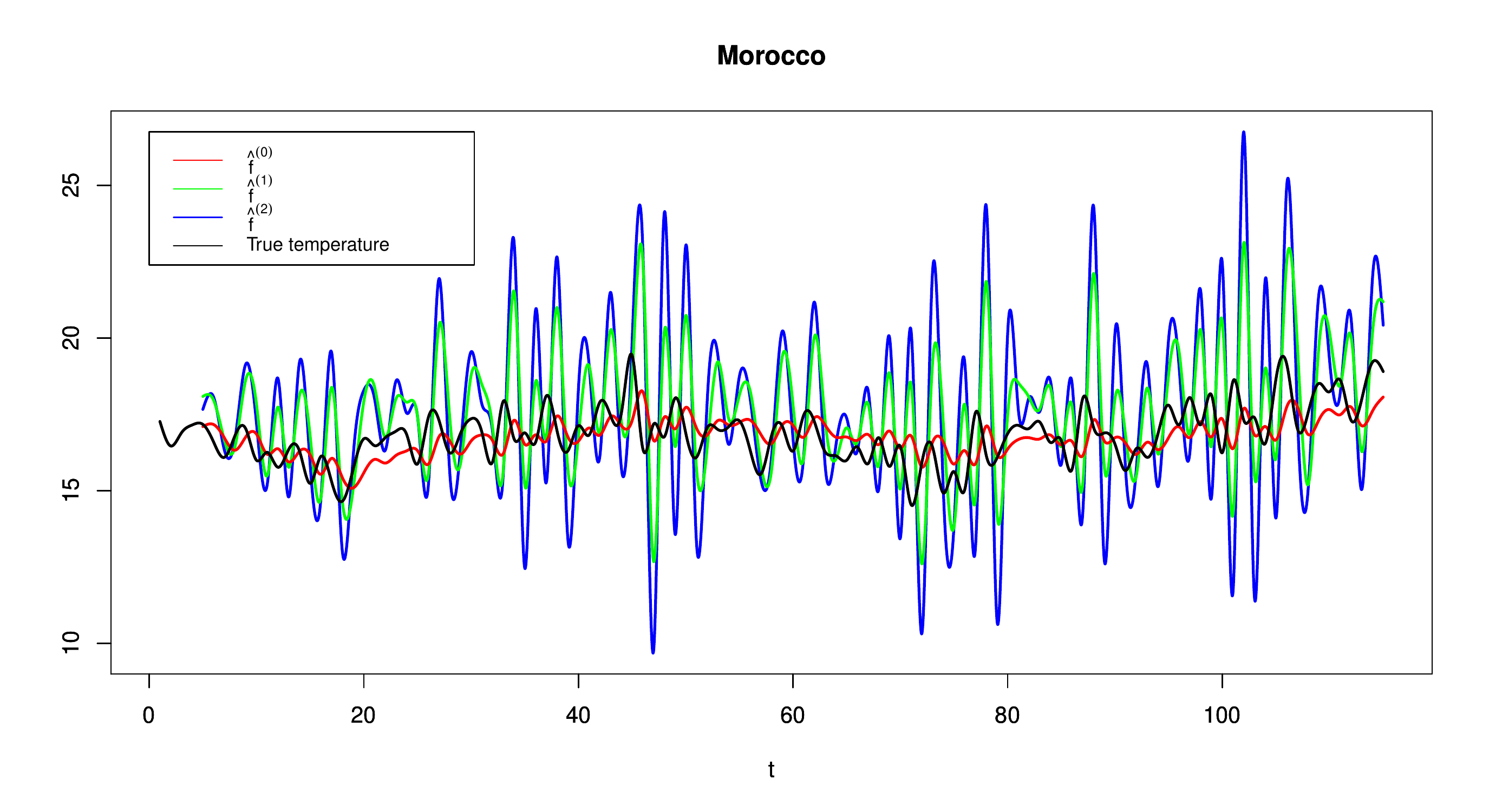}}
	\end{center}
\end{figure}

\begin{table}[H]
	\centering
	\caption{ The predictors  $\hat{f}^{(i)}(t_{n})$, $n=114$ (the mean annual temperature of the year 2015).}
	\label{pr2015}
	\begin{tabular}{|l|l|l|l|l|l|l|}
		\cline{1-7}
		Country &\multicolumn{3}{c|}{\textbf{France}}&\multicolumn{3}{c|}{\textbf{Morocco}}  \\ \cline{1-7}
		Kernel & $\K^{(0)}$ & $\K^{(1)}$& $\K^{(2)}$ & $\K^{(0)}$&$\K^{(1)}$& $\K^{(2)}$ \\ \cline{1-7}
		Prediction &13.17656 &15.48813 &15.61992 &  18.06526 &21.18307 &20.41619\\ \cline{1-7}
		True temperature & \multicolumn{3}{c|}{13.5} & \multicolumn{3}{c|}{18.9008} \\ \cline{1-7}
	\end{tabular}
	
\end{table}

\begin{table}[H]
	\centering
	\caption{ The predictors $\hat{f}^{(i)}(t_{n+1})$, $n=114$ (the mean annual temperature of the year 2016).}
	\label{pr2016}
	\begin{tabular}{|l|l|l|l|l|l|l|}
		\cline{1-7}
		Country &\multicolumn{3}{c|}{\textbf{France}}&\multicolumn{3}{c|}{\textbf{Morocco}}  \\ \cline{1-7}
		Kernel & $\K^{(0)}$ & $\K^{(1)}$& $\K^{(2)}$ & $\K^{(0)}$&$\K^{(1)}$& $\K^{(2)}$ \\ \cline{1-7}
		Prediction &12.91553 &12.54049 &11.40698 & 17.86737 &18.99740 &18.49113\\ \cline{1-7}
	\end{tabular}
	
\end{table}

\begin{remark}
	Table \ref{err} shows that the kernel $\K^{(0)}$ wins against $\K^{(1)}$ and $\K^{(2)}$ with respect to the three kernel selection criteria.
\end{remark}

\begin{figure}[H]
	\caption{The $w_1^*$, $\hdots$, $w_n^*$ of (\ref{optimalpredictionunderconstraint}) for the kernels $\K^{(0)}$, $\K^{(1)}$, $\K^{(2)}$.}
	\label{w}
	\begin{center}
		{\includegraphics[ width=11.5cm]{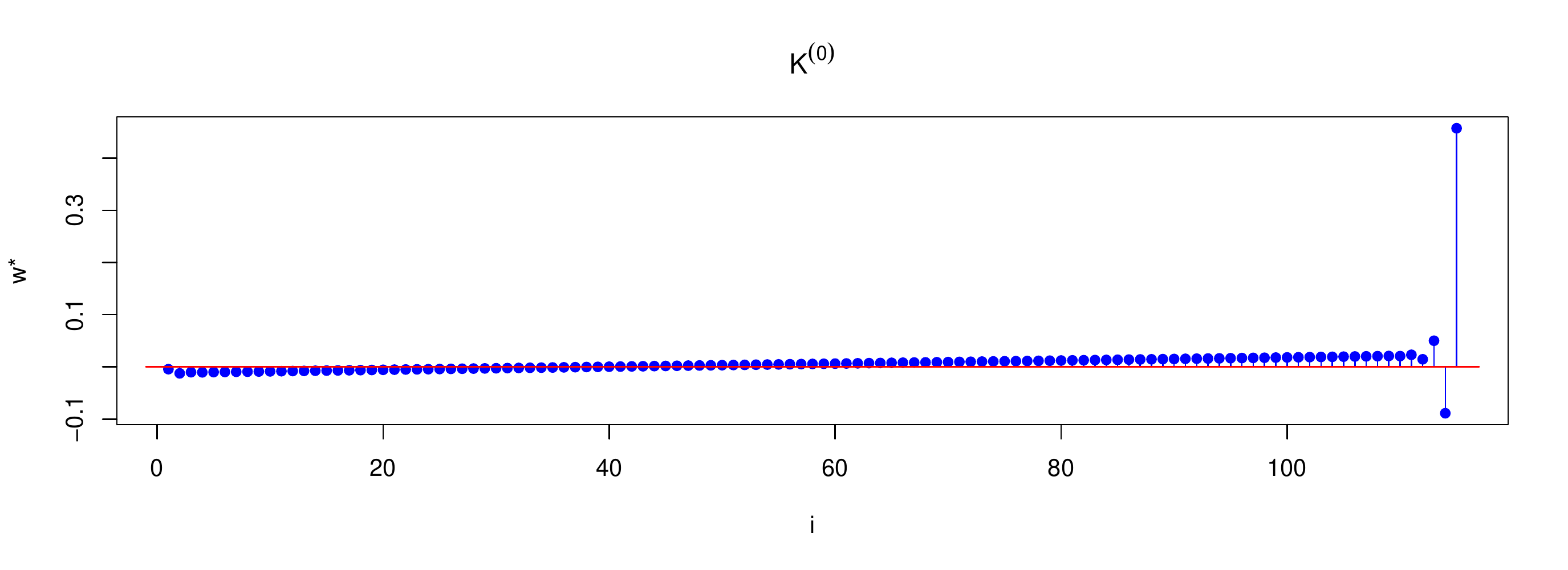}}
		{\includegraphics[ width=11.5cm]{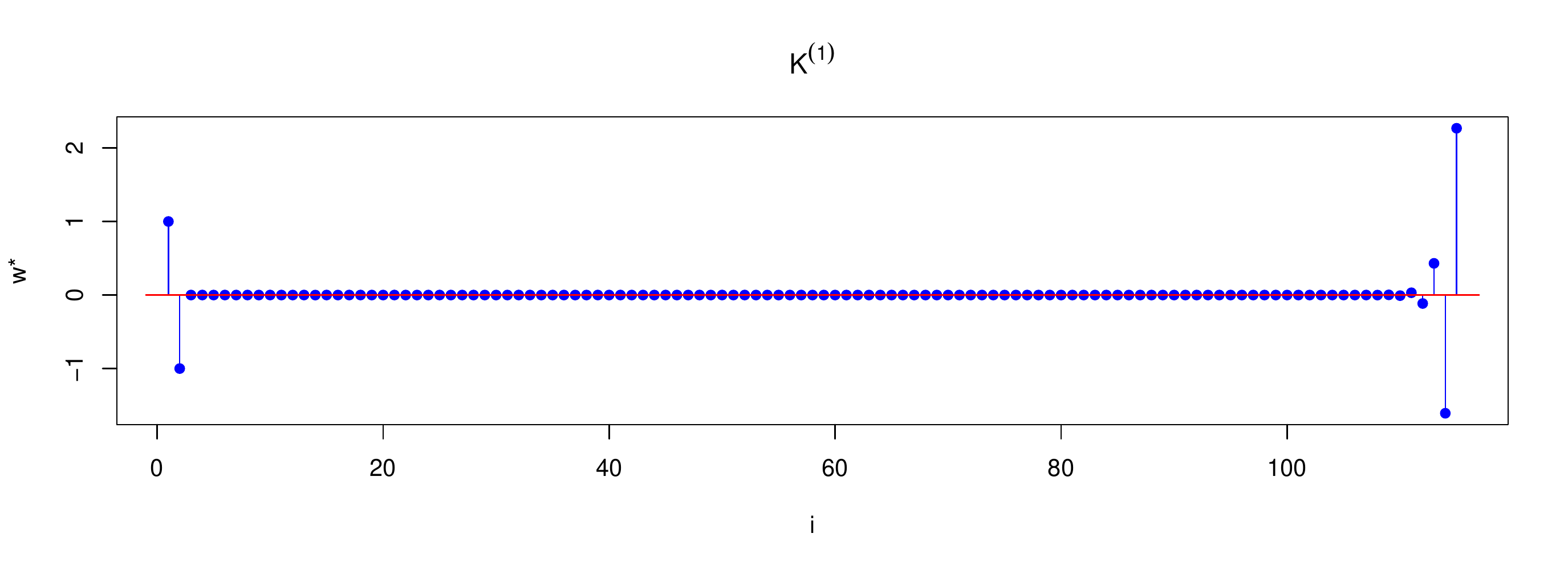}}
		{\includegraphics[ width=11.5cm]{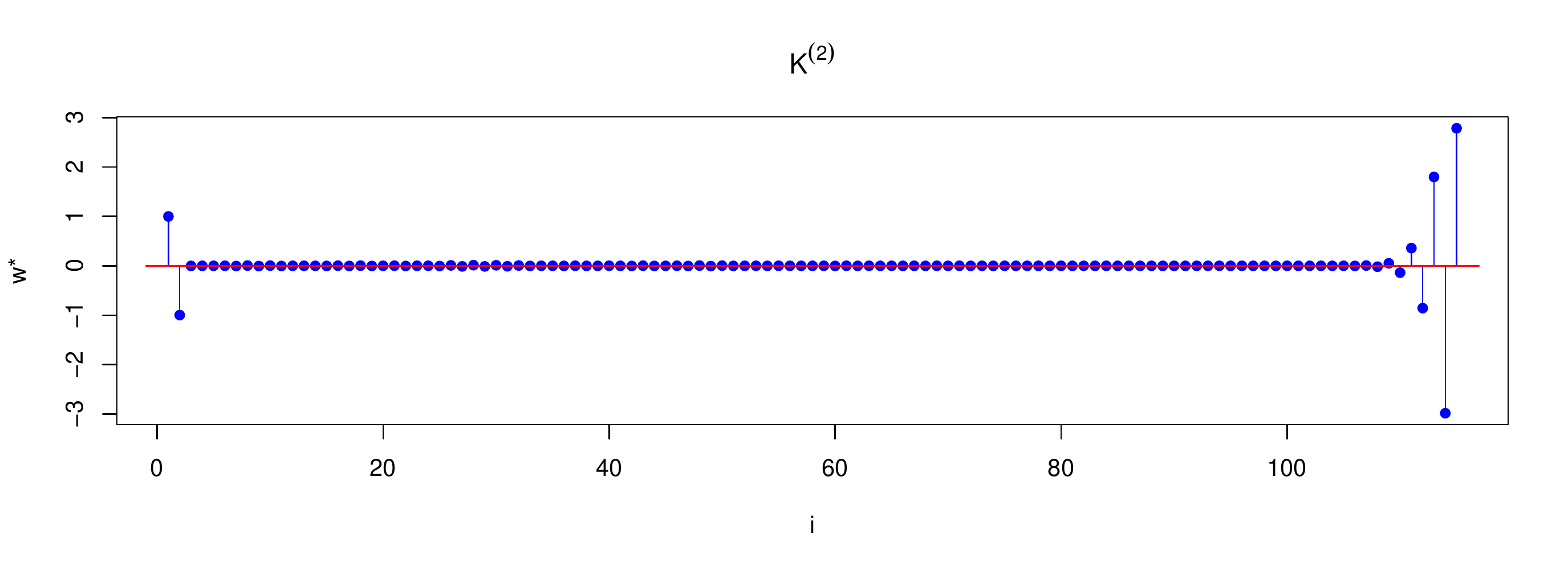}}
	\end{center}
\end{figure}

\subsection{Concluding remarks} 
The numerical results shows the three kernel selection criteria are stable, form Table \ref{err} we have that the best kernel among the three kernels is $\K^{(0)}$ w.r.t. all the three criteria for both France and Morocco data. Moreover, the representation of the splines (Figure \ref{data s}) shows that too.

From Table \ref{err} and Figure \ref{data s} we have that the kernel $\K^{(1)}$ wins against $\K^{(2)}$. 
Considering the second derivative $(u_2,\hdots,u_n)$
as Gaussian with the covariance matrix $Q^{-1}$ is a good stochastic modelization, at least is better than the assumption that $(u_2,\hdots,u_n)$
as Gaussian with the covariance matrix $Q$. 
Equivalently measuring the worst error in the unit ball using the norm 
$\|Q^{1/2}\u\|$ is better than the norm $\|Q^{-1/2}\u\|$.


\end{document}